\documentclass[%
 reprint,
superscriptaddress,
amsmath,amssymb,
aps,
pra,
]{revtex4-2}

\usepackage{mathpazo}
\usepackage{amssymb}
\usepackage{amsmath}
\usepackage{mathtools}
\usepackage{relsize}
\usepackage{hyperref}
\usepackage{amsthm}
\usepackage{physics}
\usepackage{mathrsfs}
\usepackage{braket}
\usepackage{dsfont}
\usepackage[normalem]{ulem}
\hypersetup{colorlinks=true,citecolor=blue,linkcolor=blue,filecolor=blue,urlcolor=blue}
\usepackage{matlab-prettifier}

\usepackage{float} %For forcing figures somewhere
\allowdisplaybreaks %For making calculations go over multiple pages
\usepackage{hyperref} %citations
\usepackage{mwe}
\usepackage{listing}
\usepackage[en-US]{datetime2}
\usepackage{placeins}
\usepackage{capt-of}
\usepackage{pgf,tikz}
\usetikzlibrary{arrows}
\usetikzlibrary{shapes}
\usetikzlibrary{patterns}
\usepackage{enumitem}
\usepackage{subcaption}
\usepackage{comment}

\usepackage{algpseudocode}
\usepackage{algorithmicx}

\theoremstyle{definition}
\newtheorem{algorithm}{Algorithm}
\newtheorem{theorem}{Theorem}
\newtheorem{lemma}[theorem]{Lemma}

\newtheorem{definition}{Definition}

\newtheorem{corollary}[theorem]{Corollary}

\newtheorem{proposition}[theorem]{Proposition}

\newcommand{\ve}{\varepsilon}
\newcommand{\mrm}{\mathrm}
\newcommand{\mbb}{\mathbb}

\newcommand{\wt}{\widetilde}

\newcommand{\linspan}{\mrm{span}}
\renewcommand{\ol}{\overline}

\newcommand{\opvec}{\operatorname{vec}}

\renewcommand{\tr}{\operatorname{Tr}}

\renewcommand\bra[1]{{\langle{#1}|}}
\renewcommand\ket[1]{{|{#1}\rangle}}

\newcommand{\one}{\mathds{1}}

\newcommand{\cD}{\mathcal{D}}
\newcommand{\cE}{\mathcal{E}}

\newcommand{\cH}{\mathcal{H}}
\newcommand{\cI}{\mathcal{I}}

\newcommand{\cL}{\mathcal{L}}
\newcommand{\cM}{\mathcal{M}}
\newcommand{\cN}{\mathcal{N}}

\newcommand{\cV}{\mathcal{V}}

\newcommand{\id}{\textrm{id}}

\providecommand{\ignore}[1]{}

\begin{document}

\preprint{APS/123-QED}

\title{Local Distinguishability of Multipartite Orthogonal Quantum States: \\ Generalized and Simplified}

\author{Ian George}%
\thanks{These authors contributed equally.
%\href{mailto:qit.george@gmail.com}{qit.george@gmail.com}\\ \href{mailto:malhejji@unm.edu}{malhejji@unm.edu}}
%\email{qit.george@gmail.com
}
\affiliation{Centre for Quantum Technologies, 3 Science Drive 2, 117543, Singapore}%
\author{Mohammad A. Alhejji}
\thanks{These authors contributed equally.
%\\ \href{mailto:qit.george@gmail.com}{qit.george@gmail.com}\\ \href{mailto:malhejji@unm.edu}{malhejji@unm.edu}
}
%\email{malhejji@unm.edu}
\affiliation{Center for Quantum Information and Control, University of New Mexico, Albuquerque, NM 87131, USA.}%Lines break automatically or can be forced with \\

\date{\today}% It is always \today, today,

\begin{abstract}
In a seminal work [PRL85.4972], Walgate, Short, Hardy, and Vedral prove in finite dimensions that for every pair of pure multipartite orthogonal quantum states, there exists a one-way local operations and classical communication (LOCC) protocol that perfectly distinguishes the pair. We extend this result to infinite dimensions with a simpler proof. For states on \(\mathbb{C}^{d_A \times d_A} \otimes \mathbb{C}^{d_B \times d_B}\), we strengthen this existence result by constructing an \(O(d_A^2 d_B^2)\)-time algorithm that specifies such a perfect one-way LOCC protocol. Finally, we establish the equivalence between Walgate et al.'s result and the fact that the one-shot environment-assisted classical capacity of every quantum channel is at least 1 bit per channel use, thereby clarifying the literature on these notions. At the core of all of these results is the fact that every operator with vanishing trace admits a basis where its diagonal entries are all zero.
\end{abstract}

\maketitle

\section{Introduction}
\label{sec: intro}

In recent years, there has been a global effort to implement quantum information processing. Despite the progress made, the reality is that controlling quantum systems made up of many subsystems remains difficult. For this reason, near-term quantum information processing will have to rely heavily upon classical computing infrastructure. Two examples highlight this. Early quantum communication networks will be limited by the noise in the entanglement they distribute. To counteract this, users can act on their local portions of the quantum system and communicate classical information to convert the state of the system into the wanted resource state (see e.g.~\cite{das2021universal,chen2024capacitiesentanglementdistributioncentral}). Similarly, implementing a large quantum circuit is difficult, so it is often advantageous to build up large quantum circuits by connecting smaller circuits together via classical communication--- a method known as circuit knitting (see \cite{schmitt2025cutting,harrow2025optimal} and references therein).

In the above examples, spatially separated parties share control of a composite quantum system; because of the spatial separation, the parties can only manipulate their respective portions of the system and exchange classical messages. The set of all operations arising from the composition of these actions is called local operations and classical communication (LOCC). The question of what quantum dynamics can be implemented by LOCC is long-standing. While the initial motivation for this question was to better understand entanglement \cite{bennett1996mixed, nielsen1999conditions, Walgate-2000a, hayden2001asymptotic, chitambar2011local, Chitambar-2014a, chitambar2017round} (see \cite{horodecki2009quantum, gour2024resourcesquantumworld} for relevant reviews), the previous examples highlight that understanding what dynamics are implementable by LOCC is critical for implementing quantum information processing.

If we are to use LOCC to implement quantum information processing, then we should have a robust understanding of LOCC, which we do not. As identified in \cite{yamasaki2024entanglement}, one gap in our understanding of LOCC is most works assume the quantum systems acted upon are finite-dimensional, but physically relevant quantum systems often are properly modeled as infinite-dimensional \cite{nielsen2010quantum,serafini2023quantum}. Another known issue in the theory of LOCC 
is that we lack methods for constructing LOCC protocols for tasks. If one is not appealing to a pre-existing constructive result, then one is limited to recent heuristic classical algorithms for optimizing short circuits representing LOCC protocols for various tasks \cite{zhao2021practical,Yan_2025,liu2025dynamiclocccircuitsautomated,li2025optimizingloccprotocolsproduct}. While the lack of methods is partially because the structure of LOCC protocols is complex, we also do not completely understand the simplest tasks in LOCC. As often complex problems are addressed by building on our understanding of simpler cases, we likely can benefit from a robust characterization of LOCC for simple tasks.

Given the aforementioned gaps in our understanding of LOCC and the value of characterizing simple tasks, we return to one of the earliest foundational results in LOCC: each pair of finite-dimensional orthogonal pure states can be discriminated using one-way LOCC \cite{Walgate-2000a}. First, we show that the result holds for infinite-dimensional orthogonal pure states and that for any fixed nonzero probability of error, the LOCC protocol only requires the parties to communicate a finite number of bits (Section \ref{sec:discrim-result}). Second, we show that when the states are finite-dimensional, a perfect one-way LOCC protocol can be determined efficiently. To the best of our knowledge, this is the first non-heuristic efficient algorithm for constructing an LOCC protocol. Both these results generalize the original result in manners that make it more relevant to implementation. The proofs of both results hinge upon the fact an operator with vanishing trace admits a basis where its diagonal is zero \cite{Fan1984}. Motivated by the prominence of this fact in these proofs, we consider another result that uses this fact: the one-shot environment-assisted classical capacity is at least 1 \cite{gregoratti2003quantum,hayden2004correctingquantumchannelsmeasuring}. We establish this capacity result is equivalent to the claim any two orthogonal pure states are discriminable with one-way LOCC. This equivalence explains why the previous proofs of the capacity result also use the vanishing trace operator result. In total, these results generalize and simplify previous results about LOCC by identifying the relevant mathematical structure.

\section{Preliminaries}
\label{sec: prelim}

Let \(\cH_A\) and \(\cH_B\) denote two complex Hilbert spaces to be used for the description of two quantum-physical systems \(A\) and \(B\). To describe the composite system \(AB\), we use the tensor product Hilbert space \(\cH_{AB} := \cH_A \otimes \cH_B\)~\cite[Section II.4]{reed-simon1}. Each product vector \(\ket{v}_A \otimes \ket{w}_B \in \cH_{AB}\) may be put in correspondence with an operator \(S_{\ket{v} \otimes \ket{w}}\), from the dual space \(\cH_A^*\) to \(\cH_B\), defined by the action 
\begin{gather}
\label{eq: product action}
\bra{a}_A \mapsto S_{\ket{v} \otimes \ket{w}}( \bra{a}_A) := \braket{a|v} \ket{w}_B,
\end{gather}
for all $ \bra{a}_A \in \cH_A^*$. By linearly extending this correspondence to all of \(\cH_{AB}\), we obtain a canonical isomorphism between \(\cH_{AB}\) and the space of Hilbert-Schmidt operators from \(\cH_A^{\ast}\) to \(\cH_B\)~\cite[Section 0.8]{aubrun2017alice}.

Let \(\ket{\tau}_{AB}\) denote a vector in \(\cH_{AB}\). For each \(\ket{a}_A \in \cH_A\), we define $\ket{\tau^{\vert a}}_{B} := S_{\ket{\tau}} (\bra{a}_A) \in \cH_B$. If \(\mathcal{E}_A := \{ \ket{e}_{A} \}_e\) denotes an orthonormal basis for \(\cH_A\), then the aforementioned isomorphism makes plain that \(\ket{\tau}_{AB}\) admits the linear decomposition:
\begin{align}
\label{eq:cond-state-rep}
\ket{\tau}_{AB} = \sum_{e} \ket{e}_{A} \otimes \ket{\tau^{\vert e}}_{B}. 
\end{align}
In a situation where \(\dyad{\tau}_{AB}\) is the state of \(AB\) and subsystem \(A\) is measured in the basis $\mathcal{E}_{A}$, the vector $\ket{\tau^{\vert e}}_{B}$ gives the state of the \(B\) system conditioned on outcome \(e\)~\cite[Section 6.4]{schumacher2010quantum}. Of course, this only makes sense if \(e\) occurs with positive probability. We refer to such $\ket{\tau^{\vert e}}_{B}$ as conditional states. This differs from the convention adopted elsewhere in the literature where conditional states are taken to have unit norm.

Let \(\rho_{AB}\) denote a quantum state (density operator) of the \(AB\) system. By the Born rule~\cite{holevo2001statistical}, the outcome-probabilities of a measurement on \(AB\) with a discrete set of outcomes $\cI$ are determined by a tuple of positive operators \((M_{AB}^{(i)})_{i \in \cI}\) that resolve the identity, $\sum_{i \in \cI} M^{(i)}_{AB} = \mathds{1}_{AB}$, in the sense that the probability of outcome \(i \in I\) is \(\tr(\rho M^{(i)})\). A binary discrimination protocol on \(AB\) is a quantum measurement with two outcomes. As such, its statistics are fully specified by an operator \(M_{AB}\) on \(\cH_{AB}\) that satisfies \(0\leq  M_{AB} \leq \mathds{1}_{AB}\). The discrimination protocol corresponding to \(M_{AB}\) perfectly discriminates between two quantum states \(\rho_{AB}\) and \(\sigma_{AB}\) if and only if \[|\tr( M (\rho - \sigma))| = 1.\]

Let Alice and Bob be two spatially separated parties that have joint access to the system \(AB\) such that Alice can only manipulate subsystem $A$ and Bob can only manipulate subsystem \(B\). According to common usage, a \textit{one-way LOCC} binary discrimination protocol (from Alice to Bob) consists of the following sequence of actions: Alice performs a measurement on system $A$ obtaining  some outcome $i$, Alice communicates $i$ via a classical signal to Bob, Bob performs a two-outcome measurement on system \(B\) (possibly depending on $i$), and, finally, Bob outputs a guess for the state of \(AB\). Assuming Alice's measurement has a discrete set of outcomes, operators corresponding to such binary discrimination protocols have the form 
\begin{align}
\label{eq: locc op}
\sum_{i} E_A^{(i)} \otimes F_B^{(i)},
\end{align}
where \(\sum_i E_A^{(i)} = \mathds{1}_A\), and, for each \(i\), it holds that \(E_A^{(i)} \geq 0\) and that \(0 \leq F_B^{(i)} \leq \mathds{1}_B\).

Lastly, the following fact is central to this work.
\begin{lemma}\cite[Corollary 1]{Fan1984}\label{lem:trace-zero}
Let $T$ denote a trace class operator on a separable Hilbert space $\cH$. Then $\Tr(T) = 0$ if and only if there exists an orthonormal basis $\{\ket{e}\}_{e}$ for $\cH$ such that $\bra{e}T\ket{e} = 0$ for all $e$.
\end{lemma}

The assumption that \(\cH\) is separable may be relaxed without loss. Since \(T\) is trace-class, it is compact~\cite[Section.~VI.6]{reed-simon1} and so it is the (operator norm) limit of a sequence of finite rank operators on \(\cH\)~\cite[Section 2.4]{conway1990course}. This implies that there exists a separable subspace \(\cH' \subseteq \cH\) such that \(T\cH' \subseteq \cH'\) and \(T \cH'^\perp = 0\). The restriction of \(T\) to \(\cH'\) is trace-class and has vanishing trace. Complementing the basis for \(\cH'\) arising from this lemma with a basis for \(\cH'^\perp\) yields a basis for \(\cH\) where \(T\) has a matrix presentation with zeros on the diagonal.

In finite dimensions, this lemma, known since the 1950s at the latest~\cite[Sec.~56]{Halmos2012-xv}, can be seen as a consequence of the Toeplitz-Hausdorff theorem \cite{davis1971toeplitz}. The fact that \(\tr T = 0\) implies that \(0\) lies in the convex hull of the numerical range of \(T\), and so by the theorem, there exists a unit vector \(\ket{v_1} \in \cH\) such that \(\bra{v_1} T \ket{v_1} = 0\). Next, we observe that the restriction of \((\one_{\cH} - \dyad {v_1}) T\) to \(\text{span} (\ket{v_1})^\perp\) has vanishing trace and use the theorem again to find a unit vector \(\ket{v_2} \in \text{span} (\ket{v_1})^\perp\) that satisfies \(\bra{v_2} T \ket{v_2} = 0\). We go on in this way and after a finite number of iterations we shall have obtained the sought basis.

\section{Discriminating States with One-Way LOCC}\label{sec:discrim-result}
We now generalize the main result of \cite{Walgate-2000a}.

\begin{theorem}\label{thm:walgate-et-al-thm}
For every pair of orthogonal unit vectors $\ket{\psi}_{AB},\ket{\phi}_{AB} \in \cH_{AB}$, there exists a one-way LOCC binary discrimination protocol that perfectly discriminates between the states \(\dyad{\psi}_{AB}\) and \(\dyad{\phi}_{AB}\). 
\end{theorem}

\begin{proof}
Imagine there exists a basis \(\mathcal{E}_A\) for \(\cH_A\) such that for each \(\ket{e}_A \in \mathcal{E}_A\), $\langle \phi^{\vert e} \vert \psi^{\vert e} \rangle = 0$. Then, let us denote by \(P^{\psi \vert e}_B\) the projector onto the span of $\ket{\psi^{\vert e}}_B$. If Alice performs the projective measurement $(\dyad{e}_{A})_{\ket{e}_A \in \mathcal{E}_A}$ and communicates her outcome $e$ to Bob, then he may perform the binary projective measurement $(P^{\psi \vert e}_B, \mathds{1}_B-P^{\psi \vert e}_B)$ to perfectly discriminate between the two states. In other words, it holds that
\begin{align*} 
\tr\big(\sum_{\ket{e} \in \mathcal{E}_{A}} \dyad{e}_A \otimes {P^{\psi \vert e}_B} (\dyad{\psi}_{AB} - \dyad{\phi}_{AB}) \big) = 1. 
\end{align*}
Now, our business becomes showing such a basis exists.

Consider the operator $S_{\ket{\phi}}^* S_{\ket{\psi}} : \cH_{A}^{\ast} \to \cH_{A}^{\ast}$. As the product of two Hilbert-Schmidt operators, $S_{\ket{\phi}}^* S_{\ket{\psi}}$ is trace class. It has vanishing trace because \(\langle \phi \vert \psi \rangle = \tr(S_\ket{\phi}^* S_{\ket\psi}) = 0.\) We observe, crucially, that for each $\bra{f}_A, \bra{f'}_A \in \cH_{A}^{\ast}$, it holds that 
\begin{align}\label{eq:decomp-to-cond-states}
S_{\ket{\phi}}^* S_{\ket{\psi}} (\bra{f}_A) \ket{f'}_A = \langle \phi^{\vert f'} \vert \psi^{\vert f} \rangle \ . 
\end{align} 
Here, the left hand side is the inner product of \(S_{\ket{\phi}}^* S_{\ket{\psi}} (\bra{f}_A)\) and \(\bra{f'}_A\). 
By Lemma \ref{lem:trace-zero}, there exists an orthonormal basis $\{\bra{e}_A\}_{e}$ for \(\cH_A^*\) such that \(S_{\ket{\phi}}^* S_{\ket{\psi}} (\bra{e}_A) \ket{e}_A = 0\) for all $e$. By Eq.~\eqref{eq:decomp-to-cond-states}, the dual basis \(\{\ket{e}_A\}_{e}\) satisfies $\langle \phi^{e} \vert \psi^{e} \rangle = 0$ for all $e$. 
\end{proof}

When the dimension of \(\cH_A\) is infinite, it can occur that Alice's message in the above protocol is a random variable with infinite support. This means the above protocol cannot be implemented in general if Alice is constrained to send a finite number of bits. However, for each $\ve > 0$, there exists a finite subset $G_{\ve} \subseteq \cE_{A}$ such that the probabilities \(\tr( \sum_{\ket{e} \in G_{\ve}} \dyad{e}_A \otimes \mathds{1}_B \dyad{\psi}_{AB})\) and \(\tr( \sum_{\ket{e} \in G_{\ve}} \dyad{e}_A \otimes \mathds{1}_B \dyad{\phi}_{AB})\) are both at least \( 1-\ve\). After performing her measurement, Alice may send a message to Bob that specifies the outcome when it is in $G_{\ve}$, and otherwise is a symbol that signifies that an outcome in the complement $\cE_{A} \setminus G_{\ve}$ occurred. Such a message requires no more than $\lceil \log_{2}(\vert G_{\ve} \vert) \rceil + 1$ bits. When an outcome in $G_{\ve}$ is observed, Bob is capable of perfectly discriminating between the arising two conditional states of system \(B\). We summarize these conclusions in the following corollary.

\begin{corollary}
\label{cor: finite bits disc}
For every pair of orthogonal unit vectors $\ket{\psi}_{AB},\ket{\phi}_{AB} \in \cH_{AB}$ and each \(\ve > 0\), there exists a one-way LOCC binary discrimination protocol that distinguishes the pair with success probability at least $1-\ve$ and requires a finite number of bits (depending on \(\ket{\psi}_{AB},\ket{\phi}_{AB}\) and \(\ve\)) in classical communication. 
\end{corollary}

We close this section by remarking that, as noted in \cite[Section IV.A]{Walgate-2000a} for the finite dimensional case, Theorem \ref{thm:walgate-et-al-thm} implies that it is possible to perfectly discriminate between any two infinite-dimensional multipartite orthogonal states by spatially separated parties using one-way LOCC. This is done by iteratively treating multipartite states as bipartite states. To illustrate, let us consider two tripartite orthogonal states $\ket{\psi}_{ABC}$ and $\ket{\phi}_{ABC}$. First, Alice measures system $A$ in an appropriate basis so that for each outcome $i$, the conditional states $\ket{\psi^{\vert i}}_{BC}$ and $\ket{\phi^{\vert i}_{BC}}$ are orthogonal. Then, upon receiving outcome $i$, Bob measures the $B$ system in an appropriate basis so that conditioned on each outcome $j$ of his measurement, the conditional states $\ket{\psi^{\vert i,j}}_{C}$ and $\ket{\phi^{\vert i,j}}_{C}$ are orthogonal. Thus, Charlie, the third party, can perfectly discriminate between each pair of conditional states that could arise.

\section{An efficient algorithm for finding optimal protocols in finite dimensions}
\label{sec:complexity}

The proof of Theorem~\ref{thm:walgate-et-al-thm} is an existence proof. It crucially relies on the fact that for a given zero-trace operator, \textit{there exists} a basis where the operator has zeros on the diagonal; without determining such a basis, we cannot implement the optimal one-way LOCC protocol outlined in the proof. To improve upon this situation, we present an algorithm for determining such a basis in the case where \(\cH_{AB}\) is finite-dimensional. Our algorithm has time complexity $O(d_{A}^{2}d_{B}^{2})$, takes as input classical descriptions of the two states, i.e.~their amplitudes in some product basis, and outputs a unitary and a set of states that specify Alice's and Bob's measurements. This is formally stated in Theorem \ref{thm:complexity-result}.

Before we present the technical details, we provide the high-level idea and explain how this relates to Walgate et al.'s work~\cite{Walgate-2000a}. As a first step, we reduce the task of determining a basis for Alice's measurement to the task of determining a unitary $U$ that converts a trace zero matrix encoding the overlaps of the conditional states into a matrix with zeros on the diagonal (Proposition \ref{prop:reduction-to-finding-a-unitary}). We also show how to efficiently construct this overlap matrix (Proposition \ref{prop:construct-overlap-matrix}). As a second step, we provide an efficient algorithm for determining said $U$ (Proposition \ref{prop:complexity-of-flattening}). Our first step is the same as \cite[Theorem 1]{Walgate-2000a}, but our proof is simpler, clarifies that work's observed `curious transformations' are an artifact of a standard vectorization map identity, and allows us to show the relevant mathematical object is efficient to construct. For our second step, we recast the quantum circuit argument in \cite[Section III]{Walgate-2000a} as a classical algorithm and establish efficiency of said algorithm.

We now turn to the technical details. Throughout this section, we use the notation $[k] = \{0,1,...,k-1\}$. We also assume that \(\dim(\cH_A)\) and \(\dim (\cH_{B})\) are finite, and, without loss of generality, that \(\dim(\cH_A) \leq \dim(\cH_B)\) so that we focus on an Alice to Bob one-way LOCC protocol. 

As we are interested in constructing an algorithm, we need some guaranteed structure on how the quantum states $\ket{\psi},\ket{\phi}$ are provided as input data. We take as input the amplitudes of $\ket{\psi},\ket{\phi}$ in a specific orthonormal product basis $\{\ket{i}_{A}\ket{j}_{B}\}_{i \in [d_{A}],j \in [d_{B}]}$. We make use of the `$\opvec$ mapping' defined via this product basis: 
$$\opvec(\ket{j}_{B}\bra{i}_{A}) := \ket{i}_{A}\ket{j}_{B}, \forall i \in [d_A], j \in [d_B].$$ 
This vec mapping is a (basis-dependent) isomorphism between operators from $\cH_{A}$ to $\cH_{B}$ and vectors in $\cH_{AB}$. We use it for three reasons. First, by defining $M_{\ket{\psi}} \coloneq \opvec^{-1}[\ket{\psi}]$ for each $\ket{\psi}_{AB} \in \cH_{AB}$, a direct calculation using Eq.~\eqref{eq:cond-state-rep} and the definition of the $\opvec$ mapping establishes that
\begin{align}\label{eq:vec-map-overlap-relation}
    M_{\ket{\phi}}^{\ast}M_{\ket{\psi}} = \sum_{i,i'} \langle \phi^{\vert i'} \vert \psi^{\vert i} \rangle \ket{i'}\bra{i} \ .
\end{align}
This shows $M_{\ket{\phi}}^{\ast}M_{\ket{\psi}}$ encodes the needed overlap information simply. Second, the vec mapping takes constant time to implement. This is because a multidimensional array is a list of numbers with appended data specifying how to parse when a new row starts. Implementing the vec mapping only modifies the parsing data and thus is dimension independent. Combining this fact with that multiplying a $p \times m$ matrix by an $m \times n$ matrix is $O(mnp)$, we conclude the following.
\begin{proposition}\label{prop:construct-overlap-matrix}
    Given vectors $\ket{\psi}_{AB},\ket{\phi}_{AB} \in \cH_{AB}$ expressed in the orthonormal product basis $\{\ket{i}_{A}\ket{j}_{B}\}_{i \in [d_{A}],j \in [d_{B}]}$, the matrix presentation of the operator $M_{\ket{\phi}}^{\ast}M_{\ket{\psi}}$ in the basis $\{\ket{i}_{A}\}_{i \in [d_{A}]}$ can be constructed in $O(d_{A}^{2}d_{B})$ time.
\end{proposition}
\noindent Finally, it is straightforward to establish the relationship between the unitary that makes $M_{\phi}^{\ast}M_{\psi}$ have zeros on the diagonal and the basis in which Alice should measure her system.

\begin{proposition}\label{prop:reduction-to-finding-a-unitary}
Let $U$ be a unitary on \(\cH_A\) with the property that $UM^{\ast}_{\ket{\phi}}M_{\ket{\psi}}U^{\ast}$ has zeros on the diagonal. If Alice performs the projective measurement $(U^{T}\dyad{i}\overline{U})_{i \in [d_A]}$, then for each outcome $i \in [d_{A}]$, the arising conditional states of Bob's system are orthogonal.
\end{proposition}

\begin{proof}
Recall the identity
\begin{align}\label{eq:vec-identity}
        \opvec(XYZ) = (Z^{T} \otimes X)\opvec(Y)
\end{align} 
for any matrices $X,Y,Z$ such that \(Y\) and $XYZ$ are in the domain of \(\opvec\). Using this identity with $X = \mathds{1}$, $Z = U$, and $Y = M_{\psi}$, along with the fact that \(\opvec\) is a bijection, we see that $M_{\psi}U^{\ast} = M_{\ol{U} \otimes \mathds{1} \ket{\psi}}$. Thus, $UM^{\ast}_{\ket{\phi}}M_{\ket{\psi}}U^{\ast} = M^{\ast}_{\ol{U} \otimes \mathds{1} \ket{\phi}}M_{\ol{U} \otimes  \mathds{1} \ket{\psi}}$; since this operator is assumed to have zeros on the diagonal, Eq.~\eqref{eq:vec-map-overlap-relation} implies that the conditional states of $\ol{U} \otimes \mathds{1} \ket{\psi}$ and $\ol{U} \otimes \mathds{1} \ket{\phi}$ arising from Alice measuring in the basis $\{\ket{i}_{A}\}_{i \in [d_{A}]}$ are orthogonal. This is equivalent to the claim. 
\end{proof}

We now turn to showing there is an efficient algorithm for finding the unitary $U$ needed in Proposition~\ref{prop:reduction-to-finding-a-unitary}. We begin by providing an algorithm for determining $U$ when $\cH_{A} = \mbb{C}^{2}$. This is the construction from \cite[Example 2.2.3]{Horn-Johnson-book} expressed algorithmically.
\begin{algorithm}
    Computing flattening $U$ for $\mbb{C}^{2 \times 2}$ \label{alg:qubit-unitary} \\[-3mm]
    \begin{algorithmic}[1]
    \Procedure{Uflat}{$M$}
    \State $\wt{M} \gets M - \Tr[M] \mathds{1}_{\mbb{C}^{2}}/2$
    \State $(\{\lambda_{0},\lambda_{1}\},\{\ket{w_{0}},\ket{w_{1}}\}) \gets \text{eig}(\wt{M})$
    \If{$\lambda_{0} = 0$}
        \State $\ket{u} \gets \ket{w_{0}}$
        \State $\ket{v} \gets \ket{w^{\perp}_{0}},$ where $ \braket{w^{\perp}_{0} | w_0} = 0, || w^{\perp}_{0} || = 1$
    \Else 
        \State Compute $\langle w_1 \vert w_0 \rangle$
        \State $\phi \gets \arg(\langle w_1 \vert w_0 \rangle)$
        \State $\ket{x_\pm} = e^{-i\phi}\ket{w_0} \pm \ket{w_1}$
        \State $(\ket{u},\ket{v}) \gets (\ket{x_{+}}/ || \ket{x_{+}} ||,\ket{x_{-}}/ || \ket{x_{-}} ||)$
    \EndIf
    \State Return $U \gets \begin{pmatrix} \ket{u} & \ket{v} \end{pmatrix}$.
    \EndProcedure 
    \end{algorithmic}
\end{algorithm}

\begin{proposition}\label{prop:qubit-flattening-unitary}
On input $M \in \mbb{C}^{2 \times 2}$, Algorithm \ref{alg:qubit-unitary} returns a unitary $U \in \mbb{C}^{2 \times 2}$ such that the two diagonal entries of $U^{\ast} M U$ are equal.
\end{proposition}

\begin{proof}
Denote \(\wt{M} := M - \tr(M) \mathds{1}_{\mbb{C}^2}/2\). If $\ket{u}$ and $\ket{v} \in \mbb{C}^2$ are orthonormal such that $\bra{u}\wt{M}\ket{u}=0$, then the unitary $U := \begin{pmatrix} \ket{u} & \ket{v} \end{pmatrix}$ is such that the two diagonal entries of $U^{\ast}MU$ are both equal to $\Tr[M]/2$. It thus suffices to see Algorithm \ref{alg:qubit-unitary} outputs such $\ket{u}$ and $\ket{v}$. We assume $\text{eig}(\cdot)$ returns eigenvalues $\lambda_{0},\lambda_{1}$ corresponding normalized right eigenvectors $\ket{w_{0}}$, $\ket{w_{1}}$ as can be implemented with textbook numerical methods, see e.g.~\cite{stewart2022numerical}. Note $\ket{w_{0}},\ket{w_{1}}$ need not constitute a complete eigenbasis. If $\lambda_{0} = 0$, then $\bra{w_{0}}\wt{M}\ket{w_{0}} = 0$ and so \(\ket{w_{0}}\) and any vector \(\ket{w_0^{\perp}}\) defined by \(\braket{w^\perp_0|w_{0}} = 0\) and \(|| w^\perp_0 || =1\) suffice. If $\lambda_{0} \neq 0$, then \(\ket{w_1}\) is linearly independent of \(\ket{w_0}\). A direct calculation shows $\bra{x_{+}}\wt{M}\ket{x_{+}} = \lambda_{0}2i \Im(e^{-i \phi}\langle w_1 \vert w_0 \rangle)$. However, we solved for $\phi$ such that $\langle w_1 \vert w_0 \rangle = e^{i\phi}\vert \langle w_1 \vert w_0 \rangle \vert$, so $\Im(e^{-i \phi}\langle w_1 \vert w_0 \rangle) = 0$. Thus, $\bra{x_{0}}\wt{M}\ket{x_{0}} = 0 = \bra{u}\wt{M}\ket{u}$. Moreover, $\langle u \vert v \rangle$ is proportional to $\langle x_{+} \vert x_{-} \rangle = 2\Im(e^{-i \phi}\langle w_1 \vert w_0 \rangle)$, so $\ket{u}$ and $\ket{v}$ are orthogonal.
\end{proof}

We now use the two-dimensional case as a subroutine for finding the flattening unitary for matrices in \(\mbb{C}^{d \times d}\). The way this is implemented is most easily understood by the visualization in Fig.~\ref{fig:alg-visualization}. The idea is to iteratively use Algorithm \ref{alg:qubit-unitary} on two-dimensional subspaces until all entries of the matrix are equal. By choosing the two-dimensional subspaces so that as many instances of Algorithm~1 can be performed in parallel as possible, the total flattening unitary is built using $\lceil \log_{2}(d) \rceil$ layers of unitaries as shown in Fig.~\ref{fig:alg-visualization}. This idea is formalized as Algorithm \ref{alg:computing-unitary}. In Ref.~\cite{Walgate-2000a}, the same iterative argument is made to prove the existence of the flattening unitary. In contrast, as we already know the flattening unitary exists (Lemma \ref{lem:trace-zero}), our analysis shows we can determine this unitary efficiently using a classical computer. 
We also embed the two-dimensional unitaries of an iteration into a single unitary (Line 11 of Algorithm \ref{alg:computing-unitary}). This embedding allows us to only use matrix multiplication $\log(d)$ rather than $\frac{d}{2}\log(d)$ times.

\begin{figure}
    \centering
    \includegraphics[width=\linewidth]{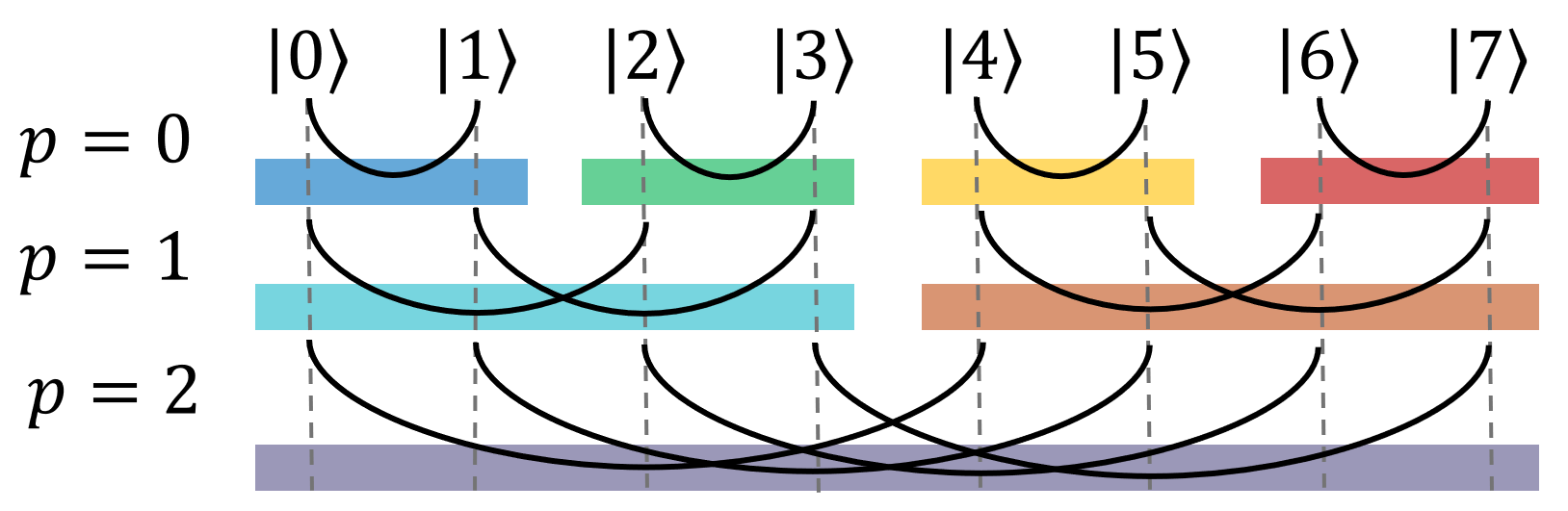}
    \caption{\footnotesize Visualization of Algorithm \ref{alg:computing-unitary} for $d$ such that $\lceil \log_2 (d) \rceil=3$. Black lines depict the pairs of computational basis states used to define the diagonal entries being flattened at iterate $p$. Blocks of colour denote diagonal elements of $M$ with the same value after iterate $p$. Because at each iterate pairs of constant value blocks are combined, the number of distinct diagonal values is cut in half each iteration. As such it takes $\lceil \log_2 (d) \rceil$ iterations (three in this case).}
    \label{fig:alg-visualization}
\end{figure} 

\begin{proposition}\label{prop:complexity-of-flattening}
On input $M \in \mbb{C}^{d \times d}$ where $d \geq 2$, Algorithm~\ref{alg:computing-unitary} returns a unitary $U$ in time $O(d^{3}\log(d))$ such that the diagonal entries of $UMU^{\ast}$ are equal. 
\end{proposition}
\begin{proof}
\textbf{Correctness:}  In Line 3, Algorithm \ref{alg:computing-unitary} embeds the matrix $M$ into \(\mathbb{C}^{2^k \times 2^k}\) where \(k = \lceil \log_2 (d) \rceil\). Lines 5-15 of the algorithm then apply the two-dimensional unitary from Algorithm \ref{alg:qubit-unitary} on two-dimensional subspaces until all the diagonal entries of $M$ are equal as we explain. At iterate $p = 0$, the algorithm constructs the unitary $U^{(0)}$ that applies the two-dimensional flattening unitary to the mutually orthogonal subspaces $\{\linspan(\ket{i},\ket{i+1})\}_{i \in \{0,2,4...,2^k-2\}}$, i.e., pairing up integers in ascending order. Thus, ${U^{(0)}}^{\ast}MU^{(0)}$ will have ascending pairs of diagonal entries of equal value and thus have at most $2^{k-1}$ distinct values. At iterate $p=1$, the algorithm constructs the unitary $U^{(1)}$ that applies the two-dimensional flattening unitary to the mutually orthogonal subspaces $\{\linspan(\ket{i},\ket{i+2})\}_{i \in \{0,1,4,...,2^k-3\}}$. Because pairs of diagonal entries of ${U^{(0)}}^{\ast}MU^{(0)}$ are already of the same value, ${{U^{(1)}}^{\ast}{U^{(0)}}^{\ast}}MU^{(0)}U^{(1)}$ will have ascending sets of $4$ diagonal entries of equal value and at most $2^{k-2}$ distinct values. Continuing in this manner, the $p^{th}$ iterate of the algorithm combines integers that are $2^{p}$ far from each other to result in at most $2^{k-p}$ distinct values on the diagonal. It thus takes $k$ iterations to guarantee a single value on the diagonal.

\textbf{Complexity:} Throughout we use $2^{k} = O(d)$. Computing $k$ is $O(\log(d))$, embedding $M$ is $O(d^{2})$, and initializing $U_{\text{tot}}$ is $O(d^{2})$. For each $p \in [k]$, there are $2^{k-1} = O(d)$ pairs of subspaces. As access to an array is $O(1)$ and $\mathsf{Uflat}$ is $O(1)$ by Proposition \ref{prop:qubit-flattening-unitary}, lines 9-11 of Algorithm \ref{alg:computing-unitary} are $O(1)$, so we may conclude lines 7-13 of Algorithm \ref{alg:computing-unitary} are $O(d)$. As updating $M$ and $U_{\text{tot}}$ are $O(d^{3})$ by direct matrix multiplication, a layer of $p \in [k]$ is $O(d^{3})$. As $k = \log(d)$, we conclude the complexity of the algorithm is $O(d^3 \log (d))$.
\end{proof}

\begin{algorithm}Computing Flattening $U$ for \(\mbb{C}^{d \times d}\) \label{alg:computing-unitary} \\[-3mm]
    \begin{algorithmic}[1]
    \Procedure{UFlatGen}{$M,d$}
    \State $k \gets \lceil \log_{2}(d) \rceil$
    \State $M \gets \begin{pmatrix} M & 0 \\ 0 & 0 \end{pmatrix} \in \mbb{C}^{2^k \times 2^k}$ 
    \State $U_{\text{tot}} \gets \mathds{1}_{\mbb{C}^{2^{k}}}$ 
    \For{$p \in [k]$}
        \State $U^{(p)} \gets \mathds{1}_{\mbb{C}^{2^{k}}}$.
        \For{$i_{1} \in [2^{k-p-1}]$}
            \For{$i_{2} \in [2^{p}]$}
                \State $i \gets i_{1}\cdot2^{p+1}+i_{2}$
                \State $j \gets i + 2^{p}$
                \State $ \begin{bmatrix} U^{(p)}_{ii} & U_{ij}^{(p)} \\ U_{ji}^{(p)} & U_{jj}^{(p)} \end{bmatrix}  = \mathsf{Uflat}\left(\begin{bmatrix} M_{ii} & M_{ij} \\ M_{ji} & M_{jj} \end{bmatrix}\right)$
            \EndFor
        \EndFor
        \State $M \gets {U^{(p)}}^{\ast}M U^{(p)}$
        \State $U_{\text{tot}} \gets {U^{(p)}}^{\ast}U_{\text{tot}}$ 
    \EndFor
    \State \textbf{return} $U_{\text{tot}}$
    \EndProcedure 
    \end{algorithmic}
\end{algorithm}

Finally, we combine the above propositions to establish the main result of this section. 

\begin{theorem}
\label{thm:complexity-result}

Given a pair of orthogonal unit vectors $\ket{\psi}_{AB},\ket{\phi}_{AB} \in \cH_{AB}$ expressed in the orthonormal product basis $\{\ket{i}_{A}\ket{j}_{B}\}_{i \in [d_{A}],j \in [d_{B}]}$, the time complexity of specifying Alice's measurement for the perfect one-way LOCC binary discrimination protocol in Theorem~\ref{thm:walgate-et-al-thm} is $O(\max\{d_{A}^{3}\log(d_{A}),d_{A}^{2}d_{B}\})$. Furthermore, specifying all of Bob's subsequent measurements is $O(d_{A}^{2}d_{B}^{2})$. In total, the entire protocol can be specified in time $O(d_{A}^{2}d_{B}^{2})$. 
\end{theorem}

\begin{proof}
    By Proposition \ref{prop:construct-overlap-matrix}, we can construct the matrix representation of $M_{\ket{\phi}}^{\ast}M_{\ket{\psi}}$ in $O(d_{A}^{2}d_{B})$ time. By Proposition \ref{prop:complexity-of-flattening}, we can construct the matrix representation of the unitary $U$ such that the matrix $UM_{\ket{\phi}}^{\ast}M_{\ket{\psi}}U^{\ast}$ has zeros on the diagonal in time $O(d_{A}^{3}\log(d_{A}))$. By Proposition \ref{prop:reduction-to-finding-a-unitary}, Alice's measurement is then $\{U^{T}\dyad{i}\ol{U}\}_{i \in [d_{A}]}$.

    To further specify Bob's measurement, note that upon receiving outcome $i$ from Alice, Bob's measurement is to project on the space spanned by $\ket{\psi^{\vert \ol{U}\ket{i}}}=\bra{i}\ol{U}_A \otimes \mathds{1}_{B}\ket{\psi}$ or its orthogonal complement. Initializing the matrix $\ol{U} \otimes I_{B}$ is $O(d_{A}^{2}d_{B}^{2})$. By direct matrix multiplication, computing $\ol{U}_A \otimes \mathds{1}_{B}\ket{\psi}$ is $O(d_{A}^{2}d_{B}^{2})$. Initializing the matrix $\bra{i}_{A} \otimes \mathds{1}_{B}$ is $O(d_{A}d_{B}^{2})$ time and multiplying the terms $\left(\bra{i}_{A} \otimes \mathds{1}_{B}\right) \cdot \left(\ol{U} \otimes \mathds{1}_{B}\ket{\psi}\right)$ is $O(d_{A}d_{B}^{2})$ time. Thus, computing the set $\{\ket{\psi^{\vert \ol{U}\ket{i}}}\}_{i \in [d_{A}]}$ takes $d_{A} \cdot O(d_{A}d_{B}^{2}) = O(d_{A}^{2}d_{B}^{2})$ time. Computing $\Vert \ket{\psi^{\vert \ol{U}\ket{i}}} \Vert_{2}$ is $O(d_{B})$, so computing the normalized set $\{\ket{\psi^{\vert \ol{U}\ket{i}}}/\Vert \ket{\psi^{\vert \ol{U}\ket{i}}} \Vert_{2}\}_{i \in [d_{A}]}$ takes $O(d_{A}d_{B})$ time. Adding up the complexity of each sequential step, specifying Alice's measurement has time complexity $O(\max\{d_{A}^{3}\log(d_{A}),d_{A}^{2}d_{B}\})$ and further specifying Bob's measurement is an additional $O(d_{A}^{2}d_{B}^{2})$ time. As we assume $d_{A} \leq d_{B}$, the time complexity of the entire procedure is $O(d_{A}^{2}d_{B}^{2})$.
\end{proof}

\section{Relation to Environment-Assisted Classical Capacity}
\label{sec:env-ass-cl-cap}

We end the paper by showing the equivalence between the existence of a deterministic one-way LOCC binary discrimination protocol for orthogonal pure states (Theorem \ref{thm:walgate-et-al-thm}) and the existence of one-shot environment-assisted classical codes (Definition \ref{def:env-ass-cl-cap}) with zero error and rate of \(1\) bit per channel use. This latter statement implies that the classical capacity of environment-assisted quantum channels is at least 1 bit per channel use.

We give a brief literature review to show the value of this equivalence. Motivated by Ref.~\cite{gregoratti2004quantum}, Hayden and King proved the existence of environment-assisted classical codes achieving a zero-error rate of \(1\) bit per channel use in Ref.~\cite{hayden2004correctingquantumchannelsmeasuring}. While they note this result is implied by the Walgate et al.~result, they instead prove it using the finite-dimensional case of Lemma \ref{lem:trace-zero}. Their result is implied by the earlier result~\cite[Theorem 5]{gregoratti2003quantum}, which motivated \cite{gregoratti2004quantum}. Furthermore, the relevant part of~\cite[Theorem 5]{gregoratti2003quantum} makes use of \cite[Lemma 6]{gregoratti2003quantum}, which itself is another method of proving the finite-dimensional case of Lemma~\ref{lem:trace-zero}. This history suggests the two results are related through Lemma~\ref{lem:trace-zero}. Our aim is to understand if and why these results are related through Lemma~\ref{lem:trace-zero}.

We establish that the capacity result and the Walgate et al.~result are equivalent (Proposition~\ref{prop:equivalence}); this is because environment-assisted classical codes are a class of one-way LOCC state discrimination protocols. In other words, while the two results seem different due to terminology, when considered carefully, they both are identifying that any two orthogonal pure states are distinguishable with one-way LOCC. As such, it makes sense that both results, even when treated as distinct claims, will appeal to the same lemma.

\begin{figure}
    \centering
    \includegraphics[width=0.9\linewidth]{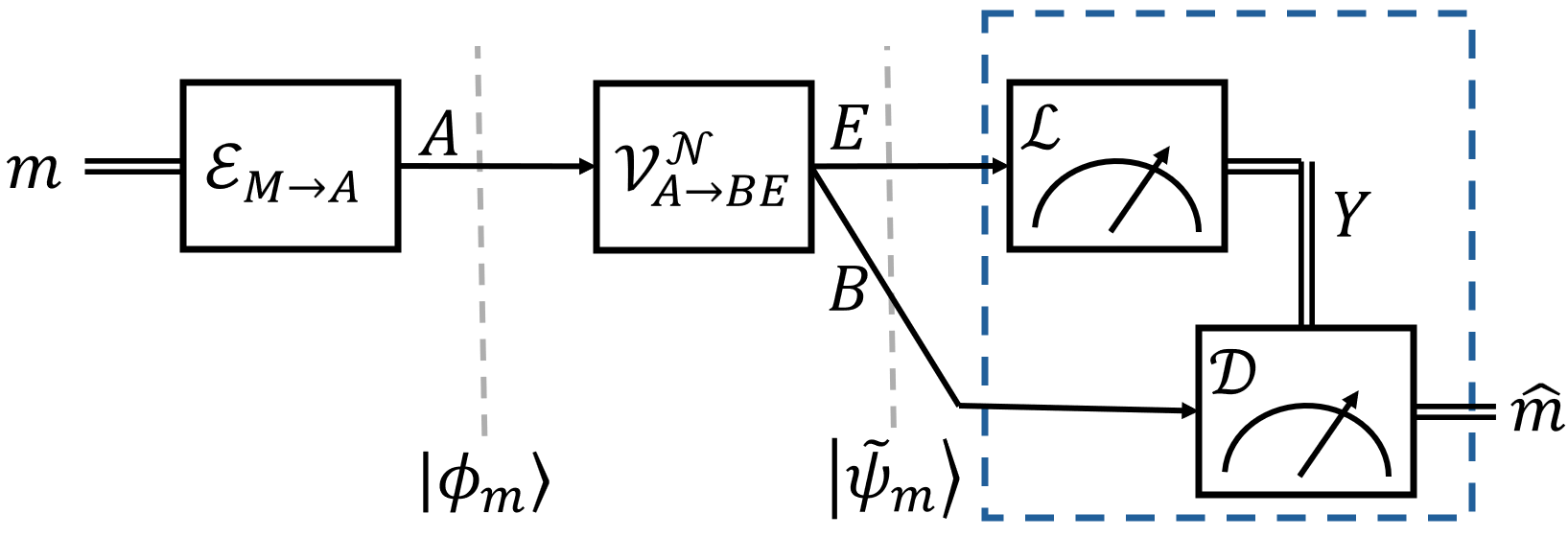}
    \caption{\footnotesize Depiction of environment-assisted classical communication and one-way LOCC discrimination. Black arrows denote quantum communication and black double-bars denote classical communication. Once the encoder \(\mathcal{E}\) is fixed, environment-assisted classical communication is a one-way LOCC state discrimination protocol with the output states generated by the isometry $\cV^{\cN}$. The one-way LOCC protocol is in the dashed blue box. For concreteness, here we depict the case where $\cE$ prepares pure states $\{\ket{\phi_{m}}\}_{m}$ reducing the problem to finding the best one-way LOCC discrimination protocol for $\{\ket{\tilde{\psi}_{m}} \coloneq \cV^{\cN}\ket{\phi_{m}}\}_{m}$.}
    \label{fig:env-ass-depiction}
\end{figure}

To prove the equivalence, we remind the reader of two basic concepts in quantum information theory, and refer them to Refs~\cite{Wilde-Book,khatri2024principlesquantumcommunicationtheory} for further background. First, the one-shot $\ve$-error classical capacity of a quantum channel $\cN_{A \to B}$ is the maximum number of bits that can be communicated over 
\begin{comment}
\end{comment}
$\cN_{A \to B}$ with error at most $\ve$. Classical information is encoded into states of a quantum system that evolve according to $\cN$, after which the information is decoded by a measurement. Second, every quantum channel $\cN_{A \to B}$ admits a (Stinespring) isometric extension, $\cV^{\cN}_{A\to BE}(\cdot) = V (\cdot) V^{\ast}$ where $V: \cH_A \rightarrow\cH_B \otimes \cH_E$ is an isometry such that 
\[\Tr_{E} \circ \cV^{\cN}_{A \to BE} = \cN_{A \to B}.\] 
The system \(E\) here is conventionally called the environment of the channel. In environment-assisted classical coding, \(E\) is controlled by a cooperative party who assists communication from \(A\) to \(B\) by performing a measurement on system \(E\) and relaying its outcome to the receiver \(B\). See Fig.~\ref{fig:env-ass-depiction} for visualization of this communication scenario.

\begin{definition}\label{def:env-ass-cl-cap}
	 Let $\cN_{A \to B}$ be a quantum channel. A one-shot environment-assisted classical code is specified by an encoder $\cE_{M \to A}$ from an alphabet $\cM$ to the input system $A$, an isometric extension $\cV^{\cN}_{A \to BE}$,  an assisting measurement $\cL_{E \to Y}$, and a decoder $\cD_{BY \to M}$. Using the shorthand \(\Phi\) for the composition \(\cD \circ (\id_{B} \otimes \cL) \circ \cV^{\cN} \circ \cE\), the average error probability of such a code is
	 \begin{equation*}
	 	p^{\text{Env}}_{\text{err}}(\cV,\cE,\cL,\cD) \coloneq \frac{1}{\vert \cM \vert} \sum_{m \in \cM} \left(1- \bra{m}\Phi(\dyad{m})\ket{m} \right) .
	 \end{equation*}
     For each $\ve \in [0,1]$, the one-shot \(\ve\)-error environment-assisted classical capacity of the quantum channel is the supremum over all codes such that the error is less than $\ve$, i.e.,
	 \begin{equation*}
	 	C^{\ve}_{\text{Env}}(\cN)  \coloneq \sup_{\cE,\cL,\cD} \{\log \vert \cM \vert : p^{\text{Env}}_{\text{err}}(\cV,\cE,\cL,\cD) \leq \ve \},
	 \end{equation*}
     where $\cV$ is any choice of isometric extension of $\cN$.
\end{definition}
\noindent We remark the capacity is independent of the choice of isometric extension because all isometric extensions are equivalent up to an isometry on the $E$ space, which can be absorbed into the assisting measurement.

\begin{proposition}\label{prop:equivalence}
The one-shot zero-error environment-assisted capacity being at least 1 for all channels whose domain has dimension \(> 1\) is equivalent to Theorem~\ref{thm:walgate-et-al-thm}.
\end{proposition}
\begin{proof}
    As has been observed previously \cite{Winter2007env}, without loss of generality, environment-assisted classical codes are one-way LOCC protocols $\cD_{BY \to M} \circ \cL_{E \to Y}$ for discriminating the states $\{\rho^{m}_{BE} \coloneq (\cV^{\cN} \circ \cE)(\dyad{m})\}$ (see Fig.~\ref{fig:env-ass-depiction} for visualization). It follows that if all pairs of orthogonal bipartite pure states are distinguishable by one-way LOCC, then the environment-assisted classical capacity of a quantum channel (whose domain has dimension \(> 1\)) must be at least 1, because the orthogonal pure states \(\cV^{\cN}\ket{0}_A\) and \(\cV^{\cN}\ket{1}_A\) can be perfectly decoded via one-way LOCC. This was noted in \cite{hayden2004correctingquantumchannelsmeasuring}. On the other hand, if there exists a quantum channel \(\cN_{A \rightarrow B}\) (with a domain of dimension \(2\) or higher) whose zero-error environment-assisted classical capacity is strictly less than 1, then it must be the case that for all orthogonal \(\ket{0}_A, \ket{1}_A\), the images \(\cV^{\cN}\ket{0}_A\) and \(\cV^{\cN}\ket{1}_A\) cannot be discriminated via one-way LOCC, so there exist some orthogonal bipartite pure states that are not discriminable with one-way LOCC. 
\end{proof}

By combining this equivalence with Theorem \ref{thm:walgate-et-al-thm}, the fact that $C^{\ve}_{\text{Env}}(\cN)$ can only increase as $\ve$ increases, and that the classical capacity of a qubit quantum channel is at most one by Holevo's theorem \cite[Chapter~11]{Wilde-Book}, we immediately obtain the following known results:
\begin{enumerate}[itemsep=0pt,itemindent=0pt]
	\item \cite[Theorem 2]{hayden2004correctingquantumchannelsmeasuring} For all quantum channels $\cN_{A \to B}$ and all $\ve \in [0,1]$, $C^{\ve}_{\text{Env}}(\cN) \geq 1$.
	\item \cite[Theorem 5]{gregoratti2003quantum}, \cite[Theorem 3]{hayden2004correctingquantumchannelsmeasuring} For any qubit channel \(\cN\), its environment-assisted classical capacity \(C_{\text{Env}}(\cN)\) is 1, i.e.,
    $$C_{\text{Env}}(\cN) \coloneq \lim_{\ve \to 0} \lim_{n \to \infty} \frac{1}{n} C_{\text{Env}}^{\ve}(\cN^{\otimes n})= 1 \ . $$
\end{enumerate}
The above discussion thus shows that one may view all of these environment-assisted capacity results being reducible to Lemma~\ref{lem:trace-zero}, as being because these capacity results are actually Theorem \ref{thm:walgate-et-al-thm} or corollaries thereof.

\section{Acknowledgments}
The authors thank Eric Chitambar and Felix Leditzky for helpful discussions about this work. IG is supported by the Ministry of Education, Singapore, through grant T2EP20124-0005.  MA acknowledges support from the US National Science Foundation Grant PHY-2116246.

\bibliography{References.bib}

\end{document}